\documentclass[10pt, conference, letterpaper]{IEEEtran}

\usepackage{amsmath} 
\usepackage{amsfonts} 
\usepackage{amssymb}
\usepackage{amsthm}
\usepackage[utf8]{inputenc}
\usepackage{graphicx}
\usepackage{subfigure}
\usepackage{caption}
\usepackage{bbm}
\usepackage{csquotes}
\usepackage{enumerate}
\usepackage[numbers,sort&compress]{natbib} 
\usepackage{color}
\usepackage[normalem]{ulem}
\usepackage{numprint}

\newtheorem{theo}{Theorem}
\newtheorem{rem}{Remark}
\newtheorem{lem}{Lemma}

\IEEEoverridecommandlockouts

\begin{document}


\title{Performance Analysis of Online Social Platforms}
\author{\IEEEauthorblockN{Anastasios Giovanidis\textsuperscript{*}, Bruno Baynat\textsuperscript{*}\thanks{\textsuperscript{*} The two first authors contributed equally.}, and Antoine Vendeville}
\IEEEauthorblockA{Sorbonne University, CNRS-LIP6, Paris, France, Email: \{firstname.lastname\}@lip6.fr}}

\maketitle


\begin{abstract}
We introduce an original mathematical model to analyze the diffusion of posts within a generic online social platform. Each user of such a platform has his own Wall and Newsfeed, as well as his own self-posting and re-posting activity. 
As a main result, using our developed model, we derive in closed form the probabilities that posts originating from a given user are found on the Wall and Newsfeed of any other. These probabilities are the solution of a linear system of equations. Conditions of existence of the solution are provided, and two ways of solving the system are proposed, one using matrix inversion and another using fixed-point iteration. Comparisons with simulations show the accuracy of our model and its robustness with respect to the modeling assumptions. 
Hence, this article introduces a novel measure which allows to rank users by their influence on the social platform, by taking into account not only the social graph structure, but also the platform design, user activity (self- and re-posting), as well as competition among posts.
\end{abstract}

\section{Introduction}

Online Social Platforms (OSPs) play a major role in the way individuals communicate with each other, share news and  get informed. Today such platforms host billions of user profiles. Although OSPs differ from one another, most of them share a common structure, which allows users to post messages on their Wall and read posts of others on a separate Newsfeed. Most OSPs also permit re-posting from Newsfeed to Wall, in order to facilitate information diffusion. With each re-post (or ``share'', or ``re-tweet'') the information becomes visible to a new audience, which may choose to adopt it or not, thus spreading further the post or halting its diffusion. In this way, posts originally generated by some user circulate inside the social network \cite{EC2012}. When the post is gradually adopted by a considerable proportion of the users, we see large cascades of information appear, and we refer to such posts as ``viral'' \cite{Adamic2013}.

Understanding how information spreads through OSPs is very important as it affects the opinion of the population over several subjects of every-day social life. 
Companies want to determine the set of most influential users for better marketing of their products \cite{Kempe03}, and they would like to predict information cascades \cite{LescovecPredict}. Such research is critical also because spreading of influence can have malevolent purposes instead \cite{FakeNewsPaper18}, such as the spread of misinformation (``fake news''). 
To be able to develop defense mechanisms against such social attacks, a concrete mathematical analysis of post diffusion through OSPs is necessary.

Related literature on the topic has mainly focused on models for opinion dynamics given a social graph, but has not yet considered either the OSP structure or user activity. It would be useful to have an analytical model that explains what makes a post become viral within the platforms, how posts of different user origin compete for visibility, and what is the role of user activity. Instead, most available research on understanding cascades and post diffusion in OSPs is data-driven \cite{EC2012}, \cite{Adamic2013},  \cite{LescovecPredict}.

\subsection{Related Literature}

In most relevant research on opinion dynamics, individuals are seen as agents whose relation is described by a social graph. Each agent has a certain opinion and at each step this opinion is updated through interaction with his direct peers. Such models can be grouped into two general categories.\\
\textit{1) Dynamics with Binary opinions:} There are only two possible opinions that agents can take.  A large amount of work descends from the \textit{voter model} \cite{HoLi75}, where opinion dynamics are based on imitation. 
A variation has been studied in \cite{YOASS}, where agents with persisting opinions are included. For further extensions, see also \cite{HayelINFO18}. 
Another group of work is related to epidemic spread. An agent is ``susceptible'' when his opinion is $0$ and becomes ``infected'' when he adopts opinion $1$, through social interaction. In \cite{Kempe03} two such mechanisms for opinion updates have been proposed. 
\\
\textit{2) Dynamics with Continuous opinions:} Several works in the literature have inherited and extended the original model of DeGroot \cite{DeGroot}. In this, each agent updates his continuous opinion by forming at each step a weighted linear combination of the current opinions of his peers. Variations of this model consider the inclusion of persistent agents \cite{Emily18}. In \cite{Silva17} this update mechanism is used to formulate and solve an opinion manipulation problem. 
To account for more realistic social behavior, the authors in \cite{BaccINFO15} consider opinion dynamics where agents interact in pairs when their opinions are already close.


\textit{Data, OSPs, and Cascades:} Instead of modeling opinion dynamics, recent works rather use available data to investigate more practically how posts spread within OSPs. The authors in \cite{EC2012} describe diffusion patterns that arise in specific online domains. Data analysis of large Facebook cascades is performed in \cite{Adamic2013}. Interestingly, the authors in \cite{LescovecPredict} propose ways to predict cascade growth using machine learning tools.

\textit{User activity:}  In \cite{WhenP} the authors identify user activity as an important control tool for influence maximisation. Making extensive use of datasets, they study the appropriate times for a user to post or re-post in an OSP in order to maximise the probability of audience response.

An interesting analytical effort to relate user activity with OSP design and post diffusion is made in \cite{SmartB16}. The authors use temporal point processes to model posting and re-posting activity of a user. They highlight the importance of the Newsfeed in post propagation and map user activity to post visibility, building on the idea that a post can be adopted by a follower when it is visible on his Newsfeed and not pushed away by competing posts. Their model, however, treats only a single user Newsfeed and does not consider the dynamics of the entire social graph. Furthermore, the dynamics of the Newsfeed list are inaccurately mimicked by a FIFO queue (see also \cite{TimelinesMenache} for another interesting approach using FIFO).

\subsection{Our Contributions}
In this work we propose an analytical model for post diffusion in OSPs, which considers the entire social graph and allows users to generate new posts, or share on their Wall existing posts they find on their Newsfeed. 
The system is described in Section \ref{systemd}.
By incorporating the Wall and Newsfeed lists, we allow posts from different origin to compete for the attention of each user. 
Re-posting activity allows shared posts to further become visible to the users' followers. Re-posting thus plays the role of an information ``valve". To the best of our knowledge such analytical model is completely original. 

The model describing the generic OSP is presented in Section \ref{models}. To cope with the enormous state-space of the initial model, we introduce an accurate approximation which decomposes the state description. We further simplify the model by focusing on posts from a single user while aggregating all the rest. In steady-state, the ``aggregated'' description results in a system of linear equations for the unknown influence probabilities. Its closed-form solution is given in Th.~\ref{Main}. Conditions for the existence of a solution, and an iterative method for the solution are provided in Section \ref{closed}. \textit{The code is available in the INFOCOM ieee final version}. 
Extended numerical experiments (Section \ref{numanal}) verify the validity of our model, and highlight the importance of user activity and OSP structure in better understanding the spread of posts in OSPs. Conclusions are drawn in Section \ref{conclu}.

\section{System description}
\label{systemd}


Let us first describe a generic social network platform, such as Facebook or Twitter. A set of users generate and share some content, denoted as \textit{posts}, through the platform. Each user has a list of \textit{followers} and a list of \textit{leaders}. A user can simultaneously be follower and/or leader of others. As a follower, he (she) is interested in the content posted by his (her) leaders. With each user two lists of posts are associated, namely a \textit{Newsfeed} and a \textit{Wall}.
A user's Newsfeed is constantly fed by the content that all of his leaders post on their Walls. A user's Wall is fed (i) by his self-generated posts that draw influence from the ``outside world'', and (ii) by posts that he shares from his Newsfeed. Hence, a user's Wall is a list of self-posts and re-posts. The generic social network platform is illustrated in Figure~\ref{system}.

\begin{figure}[h!]
\centering
\includegraphics[width=0.8\linewidth]{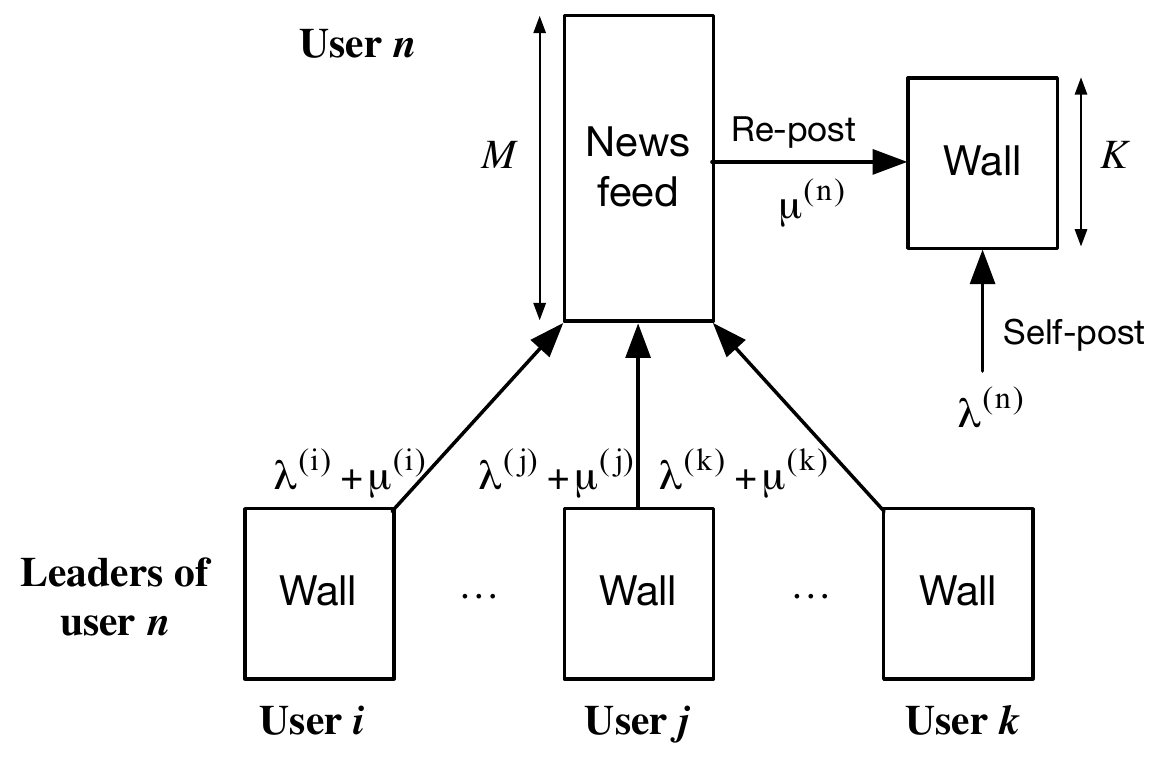}
\caption{The social platform from the point of view of user $n$.\label{system}}
\end{figure}

\subsection{Assumptions on the system and notations}


We consider a constant number $N$ of active users, forming the set $\mathcal{N}$. Users are labelled by an index $n=1,\ldots,N$. We denote by $\mathcal{F}^{(n)}$ and $\mathcal{L}^{(n)}$ the list of followers and the list of leaders of user $n$. Without loss of generality, we draw the directed Leader-graph $\mathcal{G}=(\mathcal{N},\mathcal{E})$. Each pair of nodes $(i,j)\in\mathcal{E}$, corresponds to a directed edge from $i$ to $j$, when $i$ is a leader of $j$, i.e., $i\in \mathcal{L}^{(j)}$. We denote by $\mathbf{L}$ the $N\times N$ adjacency matrix of the Leader-graph, whose coefficients are given by: $\ell_{i,j}=\mathbf{1}_{\left\{i\in \mathcal{L}^{(j)}\right\}}$,
where $\mathbf{1}_{\left\{.\right\}}$ is the indicator function. We assume that each user $n$ has at least one leader, i.e., $\mathcal{L}^{(n)} \neq \emptyset$, $\forall n$. 
The Follower-matrix is by definition $\mathbf{F}:=\mathbf{L}^T$.

The sizes of both Wall and Newsfeed are considered to be constant.
We thus fix $K$ the size of a Wall (total number of posts found on the Wall of each user) and $M$ the size of a Newsfeed. This is reasonable if we assume that only a certain number of most recent posts is considered relevant, and users don't tend to scroll down to access older posting history.


We denote by $\lambda^{(n)}$ $\mathrm{[posts/unit~time]}$ the rate with which user $n$ generates new posts on his Wall, and by $\mu^{(n)}$ the rate with which user $n$ visits his Newsfeed and selects one of the $M$ entries to re-post on his Wall (note here that each visit implies re-posting). As a result, posts arrive on the $n$-th Wall with a total rate $\lambda^{(n)}+\mu^{(n)}$ $\mathrm{[posts/unit~time]}$. Additionally, we make the assumption that content posted on a users's Wall instantaneously appears on the Newsfeeds of his followers. As a result,  the input rate of posts in the $n$-th Newsfeed, is $\sum_{j \in \mathcal{L}^{(n)}} ( \lambda^{(j)}+\mu^{(j)})$. Given that the two lists associated per user have fixed size, then with each new entry one element has to be removed from the list and replaced by the new one.
For the user activity we require $\lambda^{(n)}+\mu^{(n)}>0$, $\forall n$. 


Finally, any post originally generated by a given user $n$ takes as label the author's index $n$, and will keep this label throughout its lifespan inside the network.

\subsection{Influence metric of interest}
\label{metrics}

The aim is to estimate the influence of a specific user, say user $i$, over the entire network. In order to define the metric of interest, we first define the \textit{influence} of user $i$ on user $n$, $q_i^{(n)}$, as the steady-state probability that a post found on the Wall of user $n$ is of label $i$, i.e., has been originally created by user $i$. Note that these probabilities are performance parameters that will be the output of the developed models. They obviously satisfy $\sum_{i=1}^{N}q_i^{(n)}=1$, $\forall n$. With the above, we propose the following metric of influence,
\begin{equation}
\label{aim1}
\Psi_i = \frac{1}{N-1} \sum_{n \neq i} q_i^{(n)}.
\end{equation}
It corresponds to the average probability that a post on the Wall of any user $n\neq i$, has origin $i$. Note here that the suggested metric averages over all users in the network, but excludes the original user $i$. Other metric definitions are also possible. Since an influence score is associated with each user, the social users can be ranked by decreasing order of their influence.
%

\section{Markovian models}
\label{models}

\subsection{Modeling assumptions}
\label{assumptions}


For the analysis, we make the following assumptions:

\begin{itemize}

\item \textbf{Poisson arrivals}. For any user $n$ the generation of new posts on his Wall follows a Poisson process with rate $\lambda^{(n)}$ and the re-posting activity from his Newsfeed follows a Poisson process with rate $\mu^{(n)}$.

\item \textbf{Random selection}. When a user visits his own Newsfeed, we assume that he selects at random one of the $M$ entries to re-post on his Wall.

\item \textbf{Random eviction}. A novel entry on the Wall or Newsfeed list will push out an older entry of random position. 


\end{itemize}

Thanks to these assumptions, the resulting models developed in the following are Markovian. Indeed, all inter-arrival times between posts and re-posts are exponential and all choices are probabilistic. The random selection is consistent with common practice in real life, because it is the actual content of posts rather than the order of appearance in the Newsfeed which plays major role in the decision to re-post. The random eviction from the Wall is less realistic, since new entries are normally placed at the top of the Wall list (which would correspond to an oldest eviction policy). The validity and robustness of all these assumptions will be evaluated through simulations. As will be seen in Section~\ref{robust} they have a very limited impact on performance.

\subsection{Detailed Model}
\label{State}
The full state-description for this system is an N-tuple $(\mathbf{U}^{(1)},\ldots,\mathbf{U}^{(N)})$, where $\mathbf{U}^{(n)}=(\mathbf{x}^{(n)},\mathbf{y}^{(n)})$ is the state of user $n$ (at a given time $t$, omitted in notations for sake of clarity). $\mathbf{x}^{(n)}$ is the state of his Newsfeed and $\mathbf{y}^{(n)}$ the state of his Wall.
The random eviction and random selection assumptions allow to describe the system-state evolution without using information over the order of posts in the lists. Then, $\mathbf{x}^{(n)}=(x_1^{(n)},\ldots,x_N^{(n)})$, where $x^{(n)}_i$ counts the number of posts with user-origin $i$ found in the Newsfeed of user $n$.
Similarly,  $\mathbf{y}^{(n)}=(y_1^{(n)},\ldots,y_N^{(n)})$ where $y^{(n)}_i$ counts the number of posts with origin $i$ found in the Wall of user $n$.

With all the assumptions described in Section~\ref{assumptions}, it can be shown that the stochastic process resulting from the full state-description $(\mathbf{U}^{(1)},\ldots,\mathbf{U}^{(N)})$, is a continuous-time Markov chain model. As such, its steady-state can theoretically be solved. However, even for very small values of the system parameters the number of states will be enormous, whereas the state of a user's Newsfeed and Wall is coupled with the state of other users. As a result, any solution using, e.g., a numerical method, would be computationally intractable. For this reason we introduce in the next subsection a simple approximation that decouples the state-space and considerably reduces the solution complexity with negligible loss in precision.


Before presenting the decomposed model, it is important to understand where the coupling between states of different users appears. Consider user $n$ and focus on label $i$ posts. A leader $k$ of user $n$ will re-post from his own Newsfeed to his own Wall a post of label $i$ with probability $x_i^{(k)}/M$, due to the random selection policy. This post will appear immediately in the Newsfeed of user $n$, thus changing its state $\mathbf{y}^{(n)}$. Hence, the evolution of the state of user $n$ depends not only on his own current state and on his own activity, but also on the states and activity of all of his leaders.

\subsection{Decomposed Model}
\label{DetMod}

We now develop an approximate Markovian model which will eventually lead to a closed-form solution for the system's steady-state. The main idea is the following: for a given user $n$, the state transitions of his Newsfeed and Wall will still be a function of his own current state and activity, as well as the activity of all of his leaders. But they will not depend anymore on the current states of the user's leaders (as shown above), rather on their \textit{average probabilities in steady-state}, which at this point are unknown values.
In this way, the original full state-description can be decomposed into $2N$ independent state-descriptions resulting in $2N$ decoupled 
Markov Chains, each one associated with the Newsfeed and the Wall of a given user. The coupling between the produced $2N$ Markov Chains exists only through the unknown values of the steady-state probabilities.


\textit{For more details we refer the reader to the INFOCOM ieee version.}

\subsection{Aggregated Model}
\label{AggMod}

The state-space of the decoupled Markov chains associated with both Newsfeed and Wall of a user can still be very large. Each chain is N-dimensional. For any feasible state of the $n$-th Newsfeed, i.e., combination of posts $\mathbf{x}^{(n)}=(x_1^{(n)},\ldots,x_N^{(n)})$, it holds that $\sum_{i=1}^N x^{(n)}_i=M$. Hence, the size of its state-space corresponds to the number of ways to put $M$ undifferentiated objects in $N$ distinct boxes, and is equal to $\binom{M+N-1}{M}$.
Similarly, any state $\mathbf{y}^{(n)}=(y_1^{(n)},\ldots,y_N^{(n)})$ of the Wall of user $n$  is such that $\sum_{i=1}^N y^{(n)}_i=K$. As a result, the size of the Wall state-space is equal to $\binom{K+N-1}{K}$. These state-spaces can still be huge for large values of $N,M$, or $K$.

We now present an aggregated Markov chain model that gets around the problem by considering a reduced state-space, but \textit{without introducing any additional approximation}. As a result, the aggregated model has the same accuracy as the decomposed model presented in the previous section.

Our starting point here is the decomposed model. The idea is to focus on a particular user $i$ and develop an aggregated model that will only be able to calculate the influence of user $i$ on the entire network. Of course, one can successively apply the technique to all $i=1,\ldots,N$ in order to determine eventually the influence of all users.

We thus particularize a given user $i$ and aggregate the state-space of the system as follows. 


\textit{For more details we refer the reader to the INFOCOM ieee version.}

\section{Closed Form solution}
\label{closed}

After analytical calculations, we can obtain the following very simple \textbf{balance} expression for the Newsfeed of user $i$ and posts of origin $i$:
\begin{equation}
\label{Cpi}
p_i^{(i)}\sum_{k\in \mathcal{L}^{(i)}}\left(\lambda^{(k)}+\mu^{(k)}\right) = \sum_{k\in \mathcal{L}^{(i)}}\mu^{(k)}p_i^{(k)}
\end{equation}

Similarly, we get for the Newsfeed of any user $j\neq i$, and posts of origin $i$, the following {\textbf{balance equation}: 
\begin{equation}
\label{Cpj}
p_i^{(j)}\sum_{k \in \mathcal{L}^{(j)}}\left(\lambda^{(k)}+\mu^{(k)}\right) = \lambda^{(i)}\mathbf{1}_{\left\{i\in \mathcal{L}^{(j)}\right\}} + \sum_{k \in \mathcal{L}^{(j)}}\mu^{(k)}p_i^{(k)}.
\end{equation}

\textit{For details on the derivation of the above balance equations please refer to the INFOCOM ieee version. Note however, that these have also an intuitive interpretation. Equation (\ref{Cpi}) is an equality of two rates. On the right hand is the rate that posts of origin $i$ enter the Newsfeed of user $i$, after being selected from his leaders' Newsfeeds. On the left hand is the rate that posts of origin $i$ leave the Newsfeed of user $i$. This is just the total incoming rate, thinned by the probability that a post is of origin $i$. Remember that for our list the incoming rate is equal to the outgoing rate (no loss of posts). Hence the equation balances the ingress and egress flow of posts of origin $i$ in the Newsfeed $i$. Similar reasoning holds for equation (\ref{Cpj}). }

%

In the same fashion, the steady-state probabilities for the Wall can be directly derived from the steady-state probabilities for the Newsfeed as (remember $\lambda^{(n)}+\mu^{(n)}>0$, $\forall n$):
\begin{eqnarray}
\label{Cqi}
q_i^{(i)} & = & \frac{\lambda^{(i)}}{\lambda^{(i)}+\mu^{(i)}} + \frac{\mu^{(i)}}{\lambda^{(i)}+\mu^{(i)}}\cdot p_i^{(i)},\\
\label{Cqj}
q_i^{(j)} & = & \frac{\mu^{(j)}}{\lambda^{(j)}+\mu^{(j)}}\cdot p_i^{(j)}.
\end{eqnarray}



The above analysis gives an important structural property of the steady-state solution as a side product.

\begin{theo}[\textbf{Insensitivity in list size}]
The steady-state probabilities to find posts from user $i$ on the Wall of any user $n$ ($q_i^{(n)}$, $n=1,\ldots,N$) as well as on the Newsfeed of any user $n$ ($p_i^{(n)}$, $n=1,\ldots,N$), depend neither on the size $M$ of the Newsfeed, nor on the size $K$ of the Wall.
\end{theo}
\begin{proof}
The proof directly comes from the set of $2N$ equations (\ref{Cpi})-(\ref{Cpj}) for the Newsfeeds and (\ref{Cqi})-(\ref{Cqj}) for the Walls, that depend neither on $M$ nor on $K$.
\end{proof}

\subsection{Linear system}

We can re-write {(\ref{Cpi})-(\ref{Cpj}) and (\ref{Cqi})-(\ref{Cqj}) for posts with label $i$} in a compact form and summarize our findings as follows.

\begin{theo}[\textbf{Linear System}]
\label{th2}
The unknown column vectors $\mathbf{p}_i:=(p_i^{(1)},\ldots,p_i^{(N)})^T$ and $\mathbf{q}_i:=(q_i^{(1)},\ldots,q_i^{(N)})^T$ are the solution of the following linear system
\begin{eqnarray}
\label{LSa1}
\mathbf{p}_i & = & \mathbf{A}\cdot \mathbf{p}_i + \mathbf{b}_i\\
\label{LSb1}
\mathbf{q}_i & = & \mathbf{C}\cdot \mathbf{p}_i + \mathbf{d}_i.
\end{eqnarray}
\end{theo}

In the above, $\mathbf{A}$ and $\mathbf{C}$ are $N\times N$ matrices independent of $i$, whereas $\mathbf{b}_i$ and $\mathbf{d}_i$ are N-column vectors that depend on $i$. Hence, a standard linear system should be resolved for each $i$. The entries of the above matrices and vectors are summarised in Table \ref{T2}. It is interesting to note that $a_{j,j}=0$ for all $j$, $b_{i,i}=0$, $\mathbf{C}$ is diagonal, and also there is a unique positive $d_{i,j}$ entry for $i=j$. 

\begin{table}[ht!]
\centering
\begin{tabular}{|c|c|}
\hline
& \\
$\mathbf{A}$ & $a_{j,k}:=\frac{\mu^{(k)}}{\sum_{\ell\in\mathcal{L}^{(j)}}(\lambda^{(\ell)}+\mu^{(\ell)})}\mathbf{1}_{\left\{k\in\mathcal{L}^{(j)}\right\}}$\\
&  \\
\hline
&  \\
$\mathbf{b}_i$ & $b_{j,i} : = \frac{\lambda^{(i)}}{\sum_{\ell\in\mathcal{L}^{(j)}}(\lambda^{(\ell)}+\mu^{(\ell)})}\mathbf{1}_{\left\{i\in\mathcal{L}^{(j)}\right\}}$\\
&  \\
\hline
& \\
$\mathbf{C}$ & $c_{j,k}:=\frac{\mu^{(j)}}{\lambda^{(j)}+\mu^{(j)}}\mathbf{1}_{\left\{j=k\right\}}$\\
&  \\
\hline
 &  \\
$\mathbf{d}_i$ & $d_{j,i}:=\frac{\lambda^{(i)}}{\lambda^{(i)}+\mu^{(i)}}\mathbf{1}_{\left\{j=i\right\}}$\\
& \\
\hline
\end{tabular}
\caption{Entries for the matrices/vectors of the linear system.}
\label{T2}
\end{table}

The matrix $\mathbf{A}$ is non-negative. In addition, it is row sub-stochastic, meaning that the sum of all its rows is less or equal to $1$, with at least one row sum strictly less than $1$ (if we reasonably assume that at least one user injects self-posts). Another interesting property is that $\mathbf{A}$ is a weighted version of the Follower-matrix $\mathbf{F}=\mathbf{L}^T$, so that if $\mathbf{1}_{\left\{j\in\mathcal{F}^{(k)}\right\}} = \mathbf{1}_{\left\{k\in\mathcal{L}^{(j)}\right\}} = 0 \Rightarrow a_{j,k}=0$. There are cases however where $j$ follows $\ell$, but $a_{j,\ell}=0$ in the matrix $\mathbf{A}$, because $\mu^{(\ell)}=0$. Hence, users that never re-post from their Newsfeed 
alter the possibilities of post propagation in the graph. This is why we call $\mathbf{A}$, the \textit{propagation matrix}.

\subsection{Closed-form solution}
We would like to know under which conditions a solution to the linear system in (\ref{LSa1}) - and as a consequence (\ref{LSb1}) - exists. To this aim we first recall the following known Lemma, where $\mathbf{I}_N$ is the $N\times N$ identity matrix. It relates the solution of our system to the spectral radius of $\mathbf{A}$, denoted by $\rho(\mathbf{A})$.

\begin{lem}\cite[Chapter 6, Lemma 2.1]{BerPleNN}
\label{Lemma1}
Given a nonnegative matrix $\mathbf{T}\in\mathbb{R}_+^{N\times N}$, its spectral radius is $\rho(\mathbf{T})<1$ if and only if $(\mathbf{I}_N - \mathbf{T})^{-1}$ exists, which can be written as the series
\begin{eqnarray}
\label{series}
(\mathbf{I}_N - \mathbf{T})^{-1} & = & \sum_{n=0}^{\infty} \mathbf{T}^n\geq 0.
\end{eqnarray}
\end{lem}
From the specific structure of the non-negative matrix $\mathbf{A}$ we have the following property.

\begin{lem}
\label{Lemma2}
$\rho(\mathbf{A})\leq 1$. Strict inequality is guaranteed in the following non-exclusive non-exhaustive cases (cs):
\begin{enumerate}[(cs1)]
\item $\lambda^{(n)}>0$, $\forall n\in\mathcal{N}$.
\item For every cycle in the Leader-graph, at least one participating user has a leader $k$ with positive self-post rate.
\end{enumerate}
\end{lem}

\begin{proof}
Let us denote the row~sums of $\mathbf{A}$ by $r(j)$, $j=1\ldots N$. Then $r(j)\leq 1$ by definition from Table \ref{T2}. It is known that (\cite[Theorem 8.1.22]{HornJohn}) the following bounds are valid for the spectral radius of a non-negative matrix: $\min_{j=1}^N r(j) \leq \rho(\mathbf{A}) \leq \max_{j=1}^N r(j)$. The right-hand side in our case is $1$ and the first part is proven. 

(cs1) When $\lambda^{(n)}>0$, $\forall n$, then $\forall j$ and $k\in\mathcal{L}^{(j)}$, $a_{j,k}< \mu^{(k)} / \sum_{\ell\in\mathcal{L}^{(j)}}\mu^{(\ell)}$, so that $r(j)<1$, $\forall j$. Then the matrix is \textit{strictly} sub-stochastic, and $\rho(\mathbf{A})\leq \max_{j=1}^N r(j)<1$.

(cs2) In this case, suppose the length of a particular cycle is $\gamma>1$ and the participating nodes are $n_1,\ldots,n_{\gamma}$. Then at least one row sum $r(j)<1$, $j\in\left\{n_1,\ldots,n_{\gamma}\right\}$. By direct application of the Al'pin, Elsner, van den Dreissche bound \cite[Theorem A]{EvD08}, we conclude that $\rho(\mathbf{A})<1$. An additional condition for this bound is that $r(j)>0$, $\forall j$, which is satisfied when $\mathcal{L}^{(j)}\neq \emptyset$, $\forall j\in\mathcal{N}$ and not all leaders of some user have $\mu^{(k)}=0$.
\end{proof}

\begin{rem}
A special instance of (cs2) is when $\mathbf{A}$ is irreducible and $\lambda^{(j)}>0$ for at least one $j\in\mathcal{N}$.
\end{rem}

\begin{theo}[\textbf{Solution}]
\label{Main}
For the two cases of Lemma \ref{Lemma2}, the solution of the linear system (\ref{LSa1})-(\ref{LSb1}) is unique, and given by
\begin{eqnarray}
\label{Th1a}
\mathbf{p}_i & = & \left(\mathbf{I}_N-\mathbf{A}\right)^{-1}\mathbf{b}_i\\
\label{Th1b}
\mathbf{q}_i & = & \mathbf{C} \left(\mathbf{I}_N-\mathbf{A}\right)^{-1}\mathbf{b}_i + \mathbf{d}_i.
\end{eqnarray}

\end{theo}
\begin{proof}
Lemma \ref{Lemma2} guarantees that $\rho(\mathbf{A})<1$ in both cases, so that from Lemma $\ref{Lemma1}$ the inverse $\left(\mathbf{I}_N-\mathbf{A}\right)^{-1}\geq 0$ exists and the solution is unique.
\end{proof}

An interesting observation is that the inverse $(\mathbf{I}_N-\mathbf{A})^{-1}$ involved in the derivation of $\mathbf{p}_i$ (relation~(\ref{Th1a})) is independent of $i$. Thus, in the solution process the inverse should be calculated only once, and then applied to the expressions in (\ref{Th1a})-(\ref{Th1b}) for labels $i=1,\ldots,N$.

%

\subsection{Fixed-point algorithm}

For large $N$ it can be practically very difficult to calculate the inverse $\left(\mathbf{I}_N-\mathbf{A}\right)^{-1}$. A different way to proceed in order to solve the system (\ref{LSa1}) is to use an iterative approach.

\begin{figure*}[t]
\centering
	\subfigure{\includegraphics[scale=0.39]{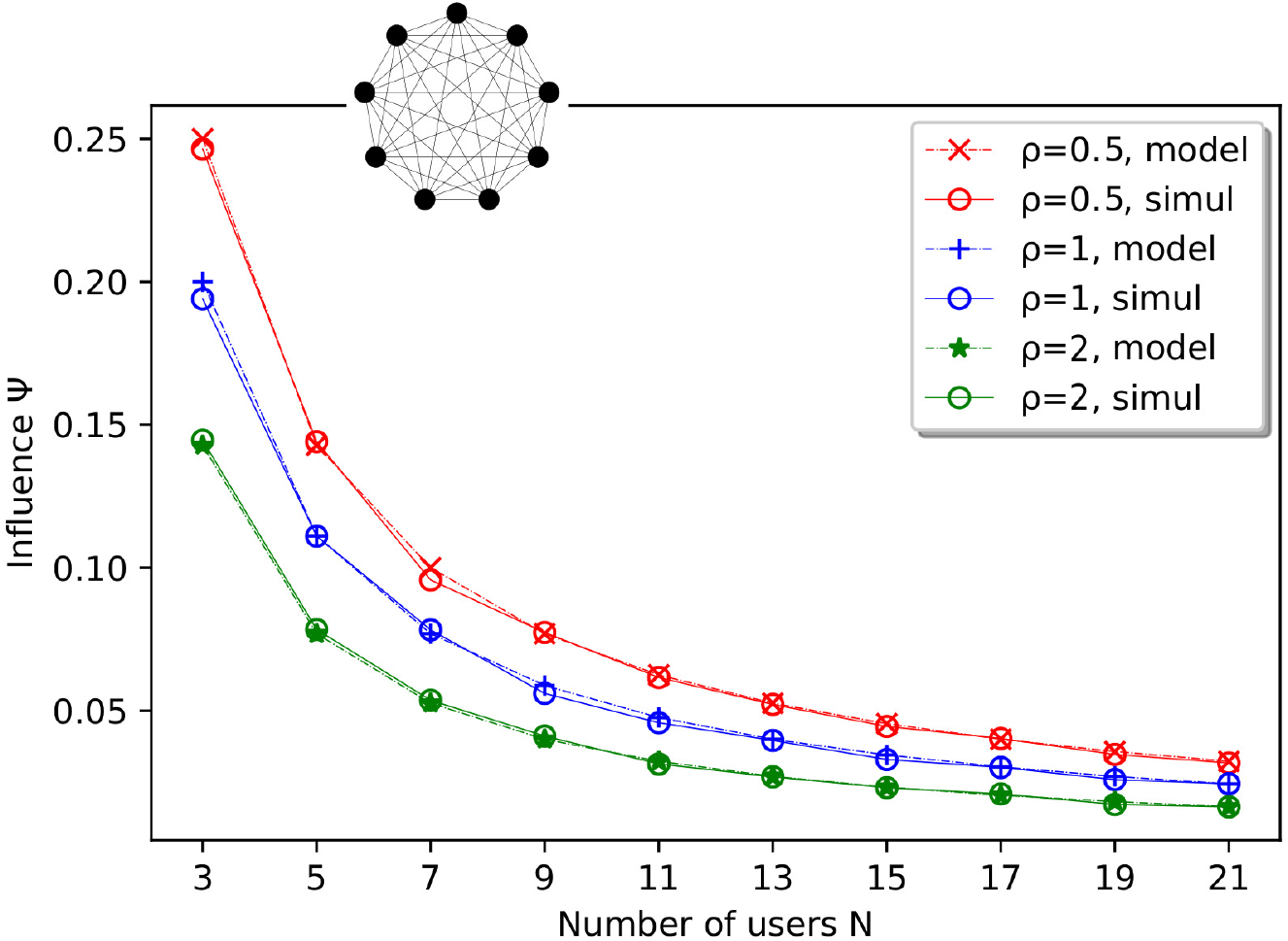}\label{validcomp}} 
	\subfigure{\includegraphics[scale=0.39]{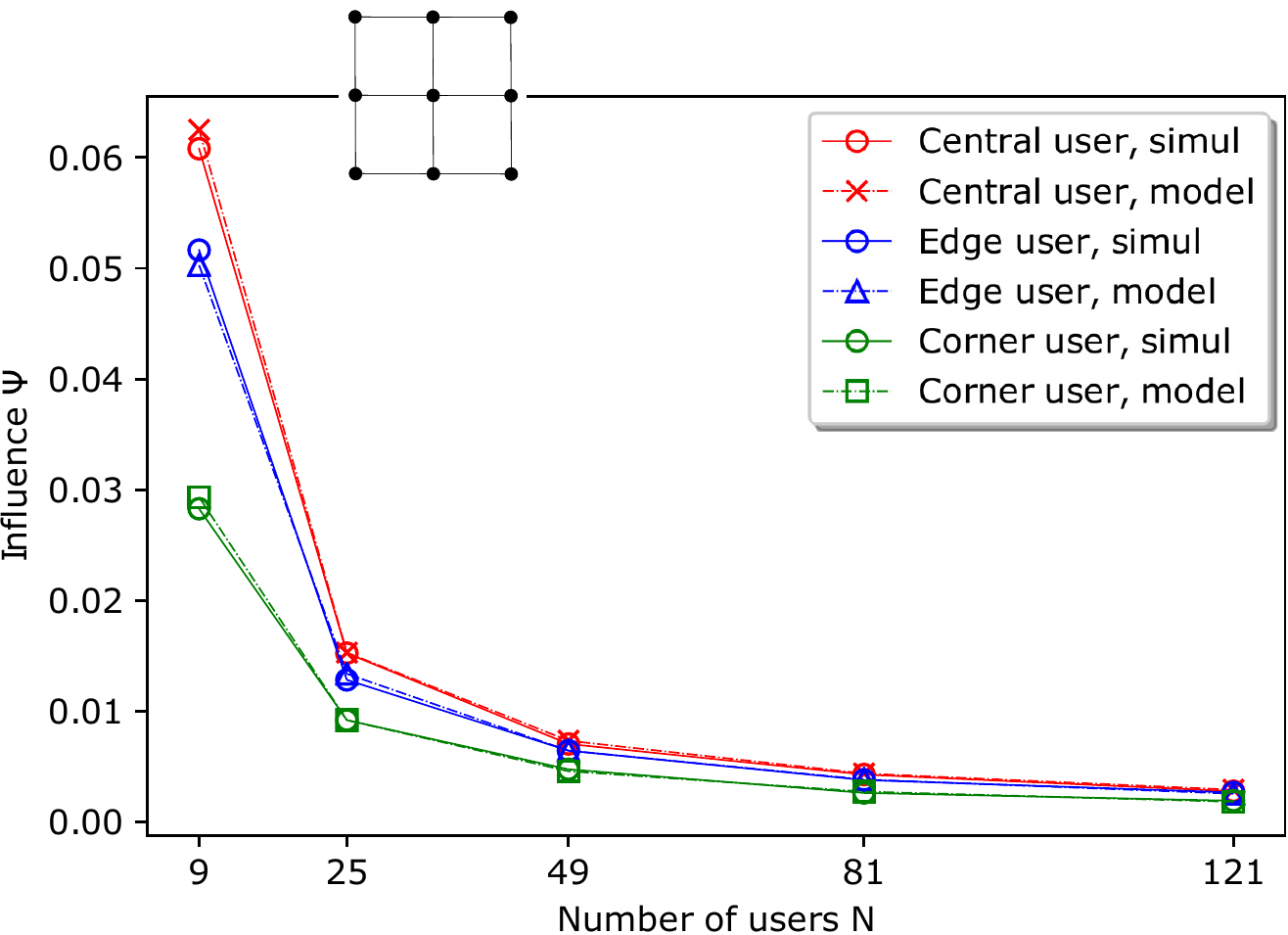}\label{validgrid}} 
	\subfigure{\includegraphics[scale=0.39]{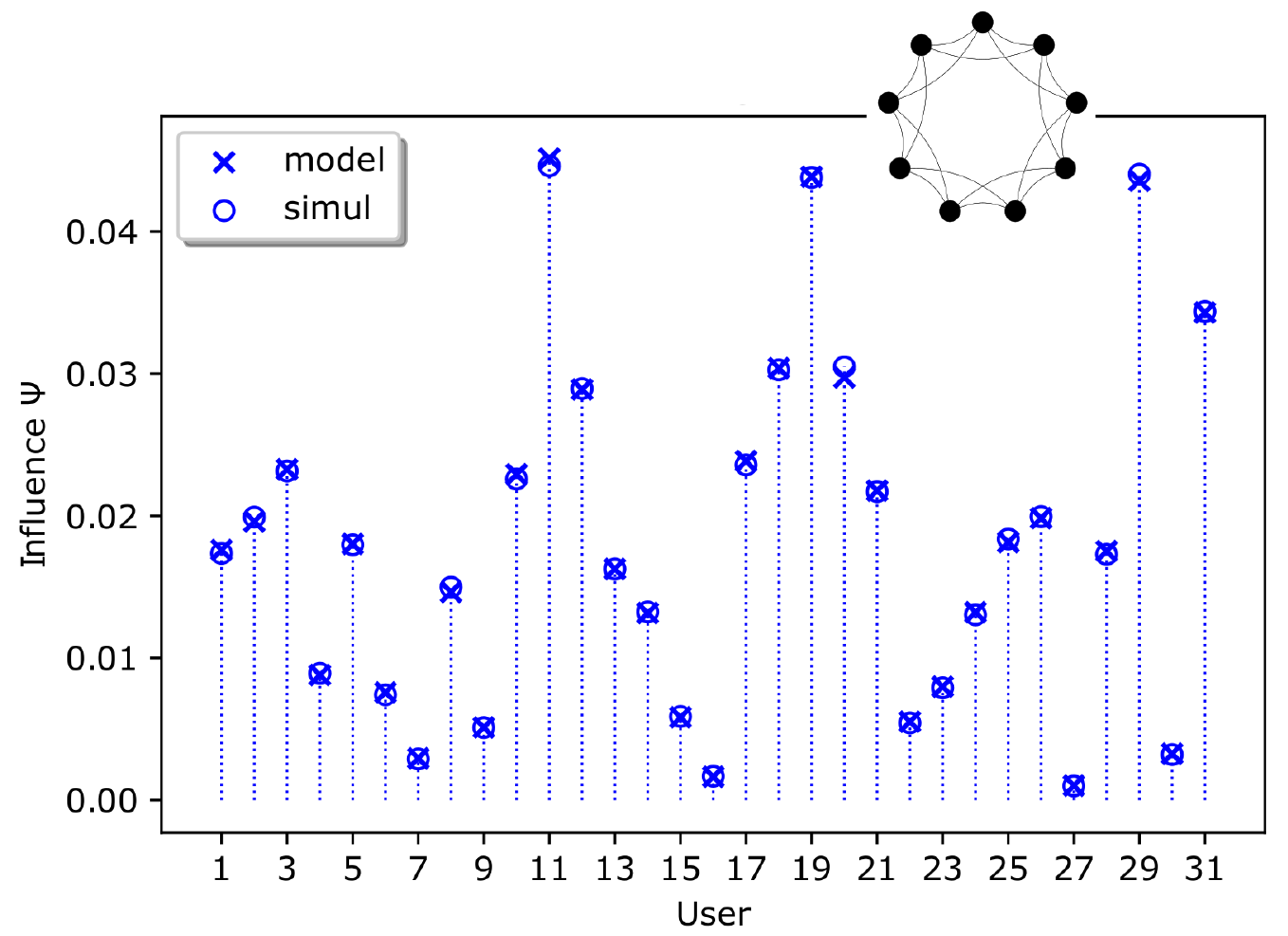}\label{validring}}
\caption{Validation of the model: (a) left: Complete graph, (b) middle: Grid graph, (c) right: Ring graph.}
\label{valid fig}
\end{figure*}


\begin{theo}
\label{t5}
For the two cases of Lemma \ref{Lemma2} and any initialization vector $\mathbf{p}_i(0)$, the discrete-time linear system (\ref{LSa}) converges towards the fixed-point solution (\ref{Th1a}) when $t\rightarrow \infty$.
\begin{equation}
\label{LSa}
\mathbf{p}_i(t) = \mathbf{A}\cdot \mathbf{p}_i(t-1) + \mathbf{b}_i\\
\end{equation}
\end{theo}

\begin{proof} We first write $\mathbf{p}_i(t)$ as a function of $\mathbf{p}_i(0)$ and $t$,
\begin{equation}
\label{Pit1}
\mathbf{p}_i(t) = \mathbf{A}^{t}\mathbf{p}_i(0) + \left(\sum_{n=0}^{t-1}\mathbf{A}^n\right)\mathbf{b}_i.\nonumber
\end{equation}
We need to find the limiting value $\mathbf{p}_i:=\lim_{t\rightarrow\infty}\mathbf{p}_i(t)$. For the two cases in Lemma \ref{Lemma2} we have $\rho(\mathbf{A})<1$, so that from \cite[pp.137--138, or Theorem 5.6.12]{HornJohn} it holds $\mathbf{A}^{\infty}:=\lim_{t\rightarrow\infty}\mathbf{A}^{t} = \mathbf{0}$. Additionally, from Lemma \ref{Lemma1} the limit of the matrix series for $t\rightarrow\infty$ converges to $\left(\mathbf{I}_N-\mathbf{A}\right)^{-1}$. Hence, the iteration converges to the solution (\ref{Th1a}), and is independent of the initialisation $\mathbf{p}_i(0)$.
\end{proof}

Note that once the Newsfeed-vector $\mathbf{p}_i:=\lim_{t\rightarrow\infty}\mathbf{p}_i(t)$ has been obtained, the Wall-vector $\mathbf{q}_i$ can be calculated from relation~(\ref{LSb1}). The performance value $\Psi_i$ is then directly derived from (\ref{aim1}).

\begin{figure*}[t]
\centering
	\subfigure{\includegraphics[scale=0.40]{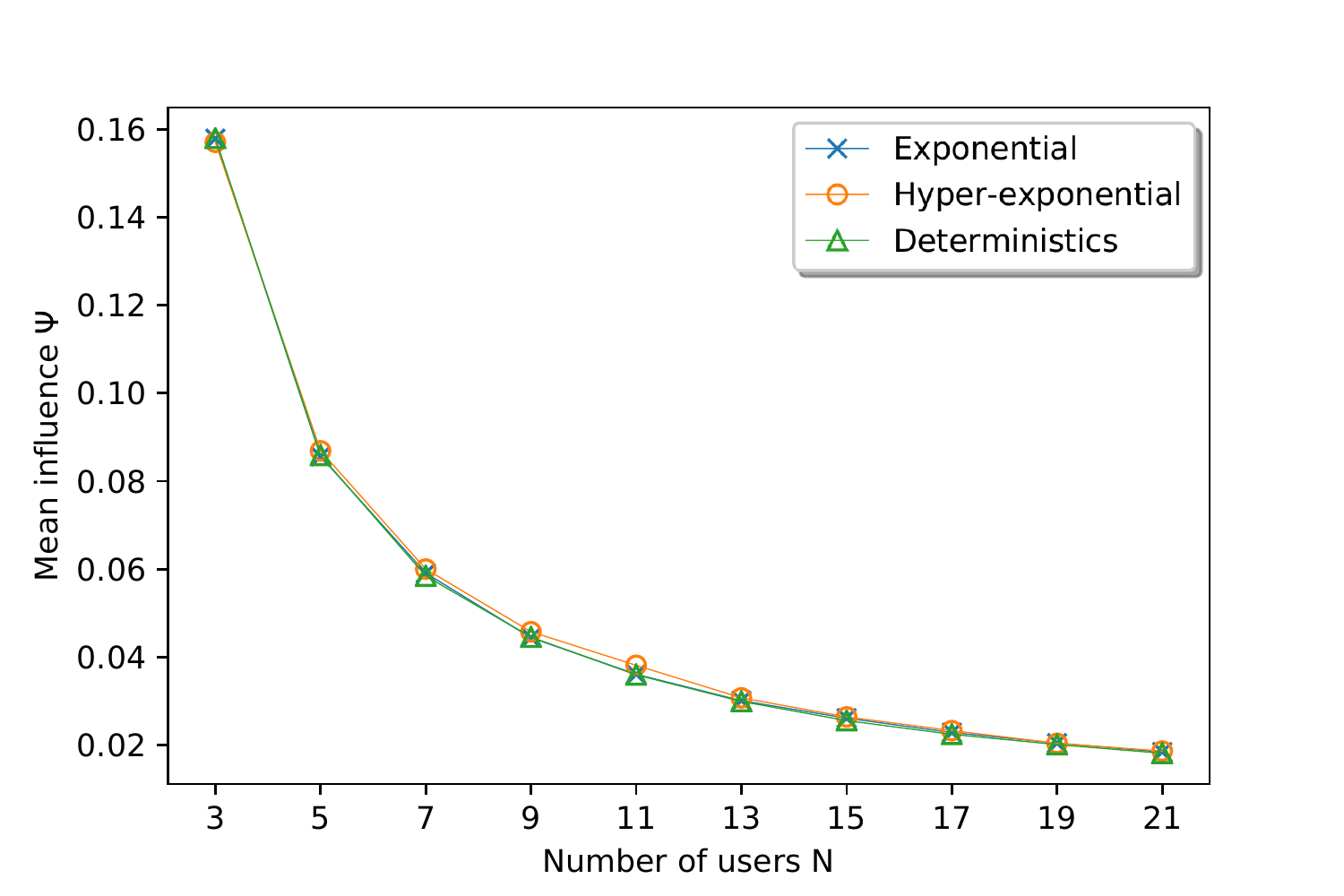}\label{robustproc}}
	\subfigure{\includegraphics[scale=0.40]{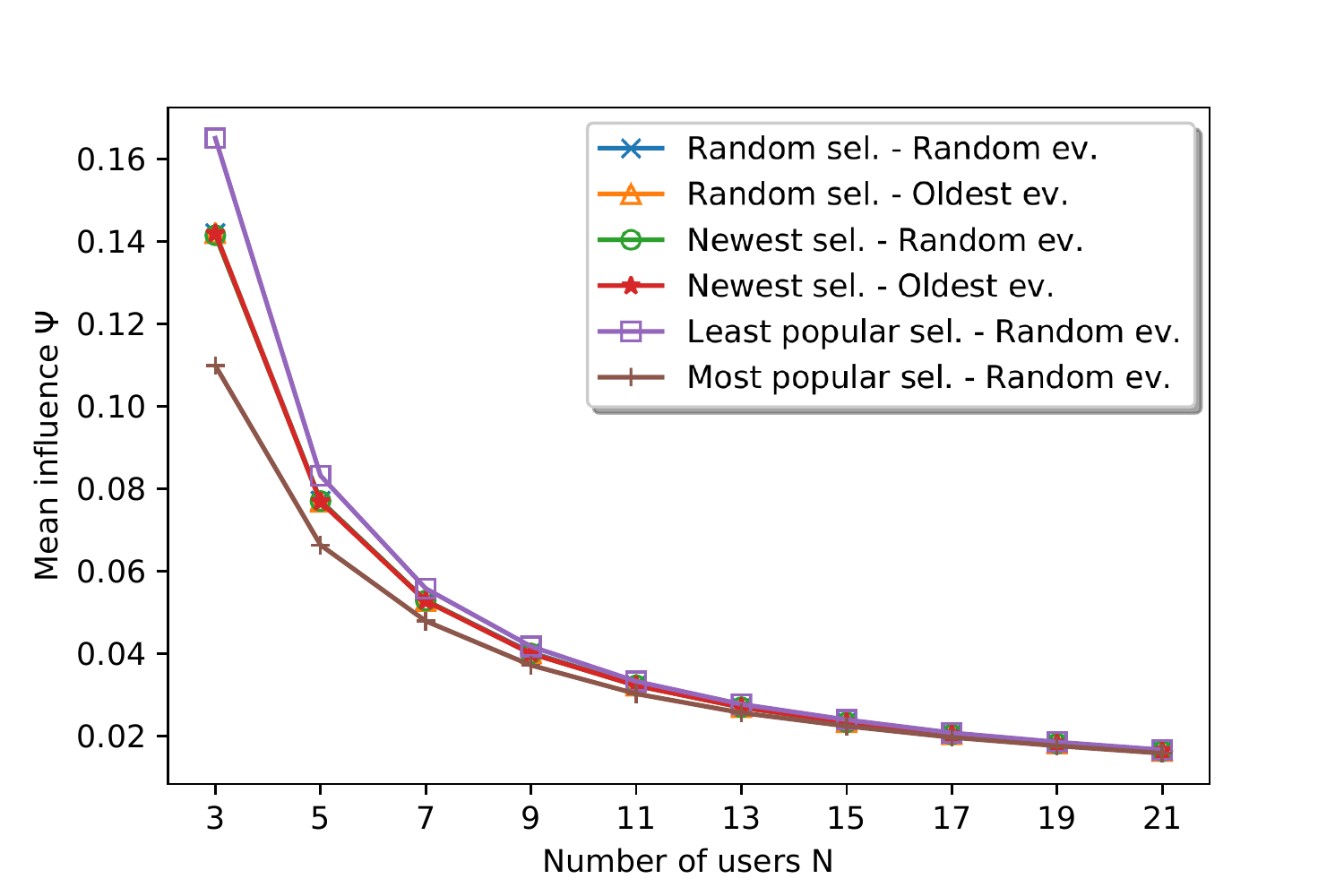}\label{robustselev}} 
\caption{Sensitivity with respect to modeling assumptions: (a) left: Inter-arrival distributions, (b) right: Selection/Eviction policies.}
\label{robust fig}
\end{figure*}

\section{Numerical Evaluation}
\label{numanal}


We have shown that the linear system in Theorem \ref{th2} can be solved either by matrix inversion, or by using a fixed-point iteration algorithm. We have programmed both methods in Python for any social-graph input and made the code freely available in the INFOCOM ieee final version. This code will be used for the numerical evaluation. Additionally, we have developed our own discrete event simulator to validate the mathematical analysis through simulation, and furthermore to evaluate the robustness of the modeling assumptions presented in Section \ref{assumptions}. Unlike our model, the simulator precisely implements the behavior of the generic OSP as described in Section~\ref{systemd}: i) The global state description consists of dynamic lists (of length $K$ for Walls and $M$ for Newsfeeds); ii) A variety of selection and eviction policies are implemented (random, in a first phase, and newest, oldest, popular later to evaluate robustness); iii) Self- and re-posts can be generated according to Poisson or other processes. As such, the simulator \textit{does not} decouple the state space, \textit{does not} estimate average probabilities, and \textit{does not} rely on Markovian assumptions.
For each simulation we set $M=20$ and $K=10$ and ran long enough simulations to reach the steady-state with small confidence intervals. 
More specifically, in all experiments, we let the simulator run for a total of $\numprint{300000}$ events (self- and re-posts). 
	
\subsection{Validation} \label{valid}

	
We compare the values of the influence metric (\ref{aim1}) resulting from the numerical model with those obtained by simulation. We use three different configurations for the user graph: complete graph, grid and ring. The results are given in Fig.~\ref{valid fig}.

\paragraph{Complete graph} In this case, each user follows all other users. All users have the same activity tuple $(\lambda,\mu)$.
As the network is totally symmetric, we plot the influence of any user over the network for three different values of $\rho:=\frac{\lambda}{\mu}$ ($0.5$, $1$ and $2$), as a function of the network size $N$. As shown on Fig.~\ref{validcomp}, there is a very good fit between model and simulation, with a maximal relative error of about $0.5\%$.
As a qualitative result, we observe that the influence of a given user decreases as the size of the network increases.  This is reasonable since the more the users in the network, the larger the competition between users to influence the Wall of eachother, thus the smaller the influence per user. Furthermore, we observe that the smaller the $\rho$, the higher the influence. This result shows that in a fully symmetric network, when everyone decreases his self-post activity, there is more space left on his Wall for being influenced by others through re-posting.
	
\paragraph{Grid graph} Depending on his position in the grid, a user may have $4$, $3$ or $2$ leaders. All users have the same activity tuple $(\lambda,\mu)$, here set to $(5,3)$. On Fig.~\ref{validgrid}, we plot the influence metric for three different types of users : the central user (with $4$ leaders), a user at the middle of an edge (with $3$ leaders) and a user at a corner (with $2$ leaders), as a function of the network size, $N=9, 25, 49, 81, 121$. We again obtain a very good match between the results of the model and simulation, with a maximal absolute error of about $10^{-3}$.
We observe that the corner user has less influence than the edge user, who in turn has less influence than the central user. This is an obvious qualitative result, considering the different numbers of followers each user has, but here we quantify the impact of the position on the chosen influence metric.
	
\paragraph{Ring graph} Users are arranged on a circle and each user $i$ has $R_i$ leaders on his right ($i+1$, ..., $i+R_i$) and $R_i$ leaders on his left ($i-1$, ..., $i-R_i$), $R_i$ being denoted as the \textit{radius} of user $i$.
For the comparison we use $N=31$ users. Each one has been given a random uniform radius in $\left\{1,...,15\right\}$ and a random uniform activity tuple $(\lambda_i,\mu_i)$ in $\left[0.1,10\right] \times \left[0.1,10\right]$. We plot the influence for each user in Fig.\ref{validring}.
Once again, we observe a very good fit between model and simulation, even for the user who has the least accurate estimation (user $``20"$ with a relative error of $2.5\%$).

\subsection{Robustness} \label{robust}
	
We further evaluate the robustness of the model with respect to the modeling assumptions: Poisson arrival processes, random selection and random eviction policies. For this purpose, we modified our simulator to take into account alternative inter-arrival distributions, as well as alternative selection and eviction policies.
In each simulation of this subsection, we chose a complete graph with a varying number $N$ of users, and we set $(\lambda,\mu)=(10,5)$.
	
\paragraph{Inter-arrival times} First we evaluate the robustness of the model with respect to the Poisson arrival assumption.
In Fig.~\ref{robustproc}, we plot the evolution of the average influence $\sum_i \Psi_i/N$ with three different inter-arrival distributions for both self-post and re-post processes of all users: exponential (corresponding to the original Poisson assumption), hyper-exponential (resulting in a process having more variability than the Poisson process) and deterministic (process with no variability). As can be seen on the figure, the three curves almost coincide. We have confirmed this observation with many different tests, highlighting the fact that the model is almost insensitive to the Poisson assumption.

\paragraph{Policies}
We programmed in our simulator alternative policies based on age and popularity. In the ``newest selection'' policy, each user always chooses the most recent post on his Newsfeed list to re-post on his Wall (instead of a random one). Hence, this policy uses extra information on the view order of posts. With the ``least popular (resp. most popular) selection'', the post to be re-posted is selected among the ones with the maximal (resp. minimal) number of re-posts. This policy also uses extra information on re-post history. ``Oldest eviction'' means that when a post has to be evicted from a Wall or a Newsfeed, it is the oldest in the list that is chosen (instead of a random one). We observe on Fig.~\ref{robustselev} that replacing random selection by newest selection and/or replacing random eviction by oldest eviction, has almost no influence on performance. Now, in the case of a selection policy based on popularity, the difference with a random policy becomes higher. This is not surprising, as selection policies based on re-post history, drastically alter the studied system. However, the effect is mainly significant for small networks. We can thus generally conclude that our model is also very robust with regard to the choice of selection and eviction policies.


\begin{figure*}[t!]
\centering
	\subfigure{\includegraphics[scale=0.4]{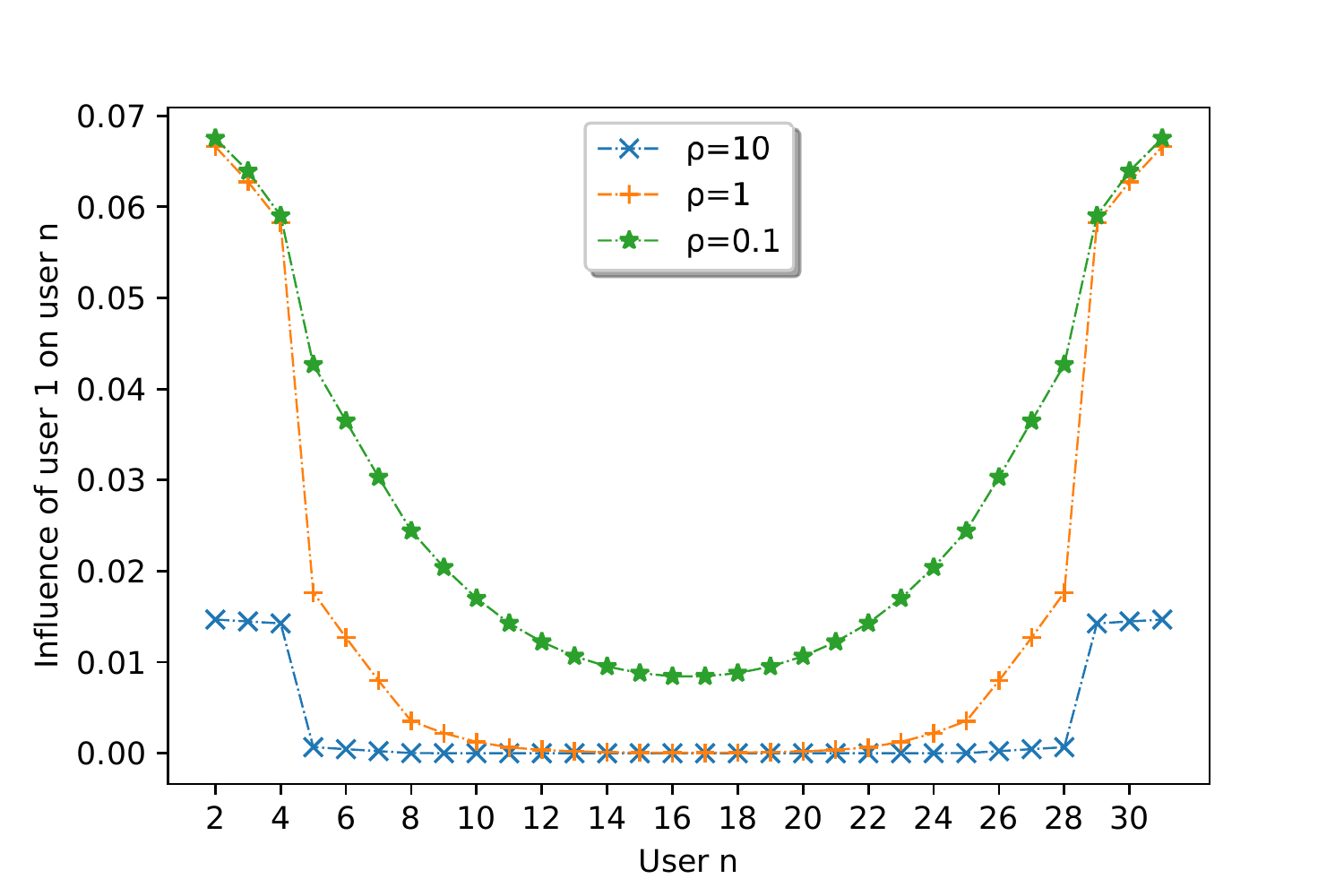}\label{ring}} 
	\subfigure{\includegraphics[scale=0.4]{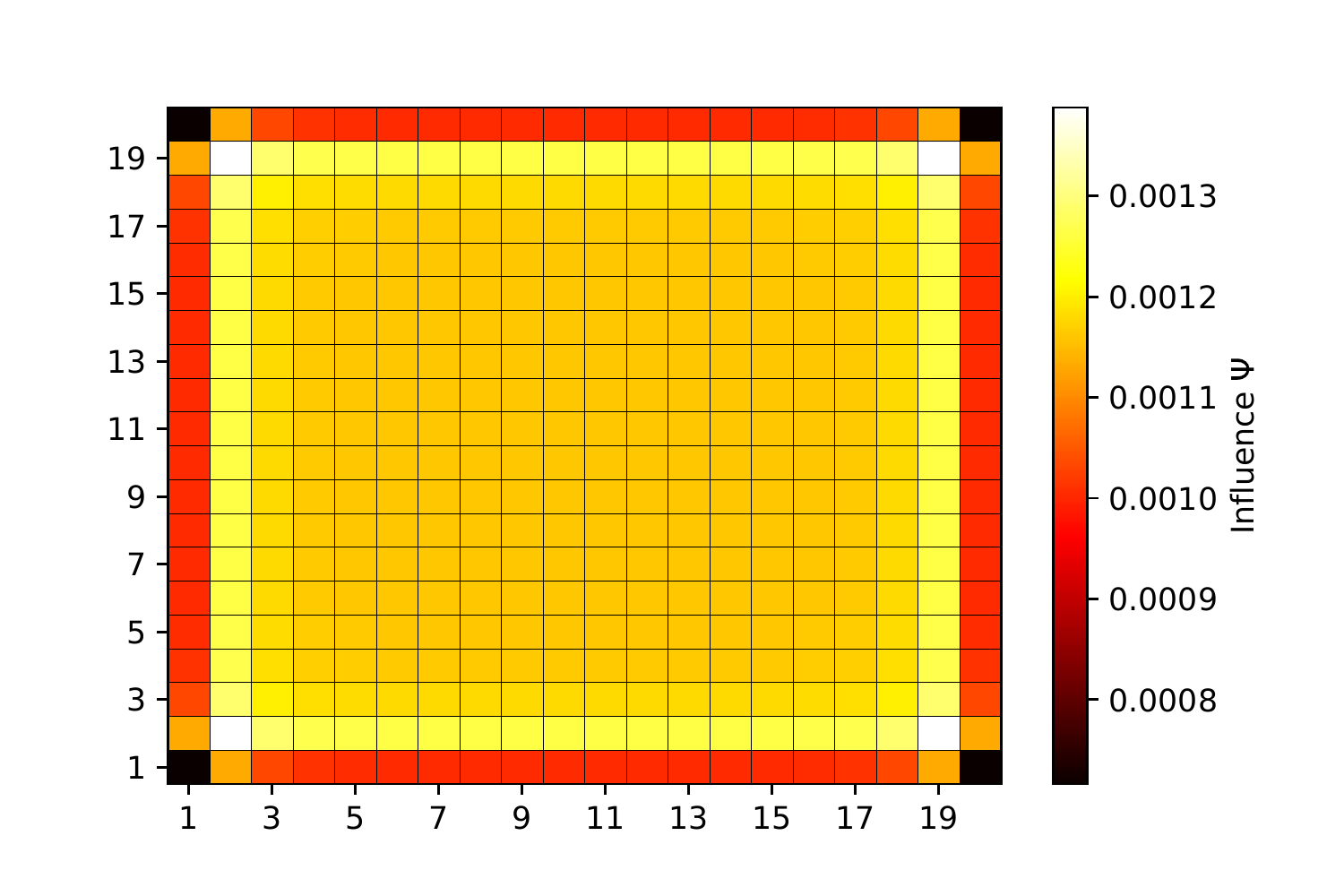}\label{grid1}}
	\subfigure{\includegraphics[scale=0.37]{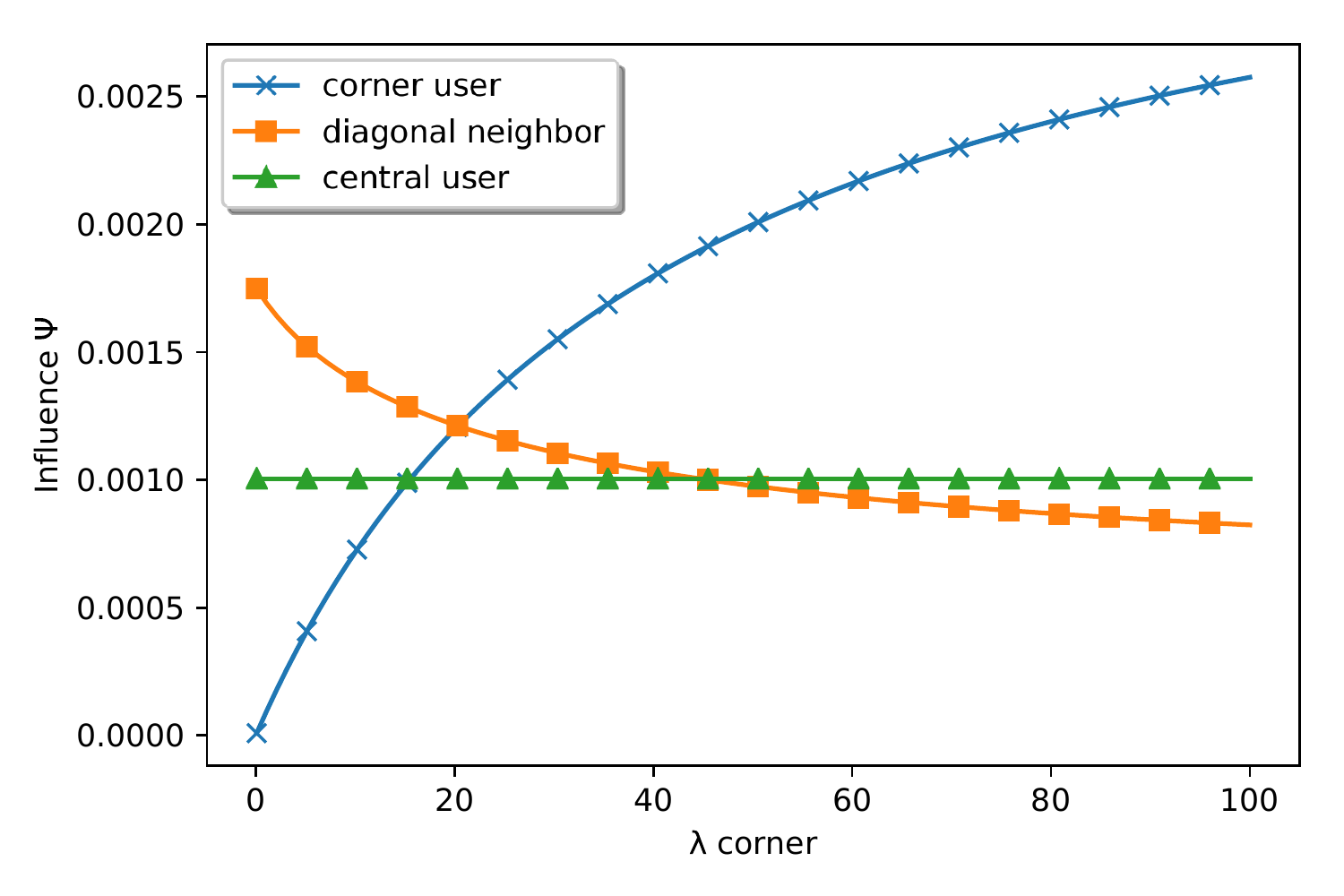}\label{grid2}}
\caption{Model exploitation: (a) left: Ring graph, (b) middle: Grid - fixed symmetric activity, (c) right: Grid - influence evolution.}
\label{perf fig}
\end{figure*}	

\subsection{Model exploitation} \label{exploit}


Having demonstrated the accuracy and the robustness of our model, we now investigate how the influence of a given user is related to his position in the graph and to the relative values of his own activity compared to the activity of other users.
	
\paragraph{Direct vs. indirect influence}

We first consider a ring with $N=31$ users, each one having the same radius $R=3$, and we plot the influence of user $``1"$ on the other users. Since $R=3$, user $``1"$ has six followers: users $\left\{``2", ``3", ``4", ``29", ``30", ``31"\right\}$. In Fig.~\ref{ring} we represent three curves corresponding to three different values of the ratio $\rho=\frac{\lambda}{\mu}$, namely $\rho=0.1$, $1$, and $10$, assuming that each user in the network has the same self-post rate $\lambda$ and the same re-post rate $\mu$. First, we observe that all the curves are symmetrical, which comes from the symmetry of the user graph. More importantly, we see that user $``1"$ has a greater influence on his direct followers, than on the other users. And obviously, the greater the distance from user $``1"$, the less influence of user $``1"$ on the considered user. Here user $``16"$ is the one that is less influenced by user $``1"$. Interestingly, the difference between the ``direct influence'', i.e., the influence of a user on one of his follower, and the ``indirect influence'', i.e., the influence of a user on a node that is not one of his follower, is closely related to $\rho$. The smaller the $\rho$, the larger the influence diffusion in the network.

\paragraph{Influence and graph position}
Here we consider a grid graph with 400 users, each one having the same activity $\left(\lambda, \mu\right) = \left(10, 10\right)$. Each square represents a user and is colored according to his influence $\Psi_i$ over the network.
As expected, the peripheral users, i.e., the users located on the outer edges, and in a more pronounced manner, the corner users are the ones with the smallest influence. This is due to the fact that these users have $3$ or $2$ direct followers, whereas all others have $4$. But contrary to intuition, the central users are not the most influential. In fact, the users with the highest influence are located in the first inner ring, i.e., at one hop from the peripheral ones. And more precisely, the most influential users are the four diagonal neighbors of the corner users. We have verified this property on different graphs. As an example, on a tree, leaves are the less influential users, whereas the parents of leaves are the most influential users. As a conclusion, in a social network where all users have the same activity, being a leader of users with few other leaders increases one's influence. Obviously this is partly due to our definition of influence, but alternative metrics have shown to follow the same trend.
	
\paragraph{Influence and activity}
We now want to see if a user with a position that gives him a low influence in the network can counterbalance his bad placement by increasing his posting activity. 
To this aim we consider again the grid of the previous subsection. We set $\left(\lambda, \mu\right) = \left(10, 10\right)$ for all users except for the south-west corner user who is given the same re-post rate, $\mu_{corner}=10$, but that can adjust his self-post rate $\lambda_{corner}$. In Fig.~\ref{grid2}, we let only $\lambda_{corner}$ vary from 0.1 to 100, keeping $(\lambda,\mu)$ and $\mu_{corner}$ fixed, and plot the activity of the considered corner user, as well as the activity of the central user and of the diagonal neighbor of the corner user (the most influential in a network with symmetric activities).

First of all,  we observe that the corner user becomes more and more influential as his posting rate increases,  and he eventually becomes the most influential user in the network. The central user is too far away from the corner user to be affected by the evolution of his posting rate, so his influence remains constant. However, observe that the raise of $\lambda_{corner}$ causes a drop in the influence of his diagonal direct neighbor.
As a conclusion, the answer is positive: one can counterbalance his bad position by increasing his self-posting activity.

\section{Conclusions}
\label{conclu}
In this work we have introduced an original mathematical model that analyzes the diffusion of posts in a generic social platform and quantifies the influence of a given user over any other within the entire network. By resolving it we have derived closed-form expressions for metrics of influence, which allow to rank users. These results constitute a novel powerful toolbox that can be further exploited to understand and design social platforms. They should be further compared and combined with network data analytics to highlight their importance. Our model can be used to derive policies for optimal user activity. But, most importantly we hope that it can lead to platform design rules that offer fair and unbiased access and post circulation for all.




\begin{thebibliography}{10}

\bibitem{EC2012}
S.~Goel, D.J.~Watts, D.G.~Goldstein.
\newblock The structure of Online Diffusion Networks.
\newblock {\em 13th ACM Conference on Electronic Commerce (EC)}, Valencia, Spain, 2012.

\bibitem{Adamic2013}
P.~Alex~Dow, L.A.~Adamic, A.~Friggeri.
\newblock The Anatomy of Large Facebook Cascades.
\newblock {\em 7th int. AAAI Conference on Weblogs and Social Media (ICWSM)}, 2013.

\bibitem{Kempe03}
D.~Kempe, J.~Kleinberg, \'{E}.~Tardos.
\newblock Maximizing the Spread of Influence Through a Social Network.
\newblock {\em ACM KDD'03, New York, NY, USA}, pp.137--146, Aug. 2003.

\bibitem{LescovecPredict}
J.~Cheng, L.A.~Adamic, P.A.~Dow, J.~Kleinberg, J.~Leskovec.
\newblock Can cascades be Predicted?
\newblock {\em 23rd int. conf. on World wide web (WWW)}, Seoul, Korea, 2014.

\bibitem{FakeNewsPaper18}
S.~Zannettou, M.~Sirivianos, J.~Blackburn, N.~Kourtellis.
\newblock The Web of False Information: Rumors, Fake News, Hoaxes, Clickbait, and Various Other Shenanigans.
\newblock {\em arxiv:1804.03461v1}, 2018.

\bibitem{HoLi75}
R.~A.~Holley, T.~M.~Liggett.
\newblock Ergodic Theorems for Weakly Interacting Infinite Systems and the Voter Model.
\newblock {\em Journal of the American Statistical Association}, Vol.3, No.4, pp.643--663, 1975.

\bibitem{YOASS}
E.~Yildiz, A.~Ozdaglar, D.~Acemoglu, A.~Saberi, A.~Scaglione.
\newblock Binary Opinion Dynamics with Stubborn Agents.
\newblock {\em ACM Transactions on Economics and Computation}, Vol.1, No.4, Article 19, pp.19:1--19:30, Dec. 2013


\bibitem{HayelINFO18}
V.~S.~Varma, I.-C.~Morarescu, Y.~Hayel.
\newblock Continuous time opinion dynamics of agents with multi-leveled opinions and binary actions.
\newblock {\em INFOCOM}, Honolulu, USA, 2018.

\bibitem{DeGroot}
M.H.~DeGroot.
\newblock Reaching a Consensus.
\newblock {\em Journal of the American Statistical Association}, Vol.69, No.345, pp.118--121, Mar. 1974.

\bibitem{Emily18}
M.~Grabisch, A.~Mandel, A.~Rusinowska, and E.~Tanimura.
\newblock Strategic Influence in Social Networks.
\newblock {\em Mathematics of Operations Research}, 43(1):29--50, 2018.

\bibitem{Silva17}
A.~Silva.
\newblock Opinion Manipulation in Social Networks.
\newblock {\em Network Games, Control, and Optimization (NETGCOOP)}, Springer, pp. 187--198, 2017.

\bibitem{BaccINFO15}
F.~Baccelli, A.~Chatterjee, S.~Vishwanath.
\newblock Pairwise stochastic bounded confidence opinion dynamics: Heavy tails and stability.
\newblock {\em INFOCOM}, pp. 1831--1839, 2015.


\bibitem{WhenP}
N.~Spasojevic, Z.~Li, A.~Rao, P.~Bhattacharyya.
\newblock When-To-Post on Social Networks.
\newblock {\em ACM KDD'15, Sydney, NSW, Australia}, pp.2127--2136, Aug. 2015

\bibitem{SmartB16}
M.~R.~Karimi, E.~Tavakoli, M.~Farajtabar, L.~Song, M.~Gomez~Rodriguez.
\newblock Smart Broadcasting: Do You Want to Be Seen?
\newblock {\em ACM KDD'16, San Francisco, CA, USA}, pp.1635--1644, Aug. 2016




\bibitem{BerPleNN}
A.~Berman, R.~J.~Plemmons.
\newblock Nonnegative Matrices in the Mathematical Sciences.
\newblock {\em SIAM Classics in Applied Mathematics; 9}, 1994.


\bibitem{EvD08}
L.~Elsner, P.~van~den~Driessche.
\newblock Bounds for the Perron root using max eigenvalues.
\newblock {\em Linear Algebra and its Applications} 428, pp. 2000--2005, 2008.

\bibitem{HornJohn}
E.~A.~Horn, C.~A.~Johnson.
\newblock Matrix Analysis.
\newblock {\em Cambridge University Press}, 1985.


\bibitem{TimelinesMenache}
A. Reiffers-Masson, E. M. Hargreaves, E. Altman, W. Caarls, D.S. Menasch\'e.
\newblock Timelines are Publisher-Driven Caches: Analyzing and Shaping Timeline Networks.
\newblock {\em SIGMETRICS Performance Evaluation Review}, 44(3): 26-29, 2016.
 



\end{thebibliography}
\end{document}